\documentclass[11pt]{article} 

\usepackage[titlenumbered,ruled]{algorithm2e} 
\usepackage[margin=1.2in]{geometry}
\usepackage{natbib} 
\usepackage{graphicx} 
\usepackage{booktabs}
\usepackage{amsmath,amsthm,amssymb} 
\usepackage{bbold} 
\usepackage{url}

\theoremstyle{definition}

\newtheorem{theorem}{Theorem}[section]
\DeclareMathOperator*{\argmin}{arg\,min}

\begin{document}

\title{The Sliding Window Discrete Fourier Transform}
\author{Lee F. Richardson \and William F. Eddy}
\date{\today}
\maketitle

\begin{abstract}
	This paper introduces a new tool for time-series analysis: the Sliding Window Discrete Fourier Transform (SWDFT). The SWDFT is especially useful for time-series with local-in-time periodic components. We define a $5$-parameter model for noiseless local periodic signals, then study the SWDFT of this model. Our study illustrates several key concepts crucial to analyzing time-series with the SWDFT, in particular Aliasing, Leakage, and Ringing. We also show how these ideas extend to $R > 1$ local periodic components, using the linearity property of the Fourier transform. Next, we propose a simple procedure for estimating the $5$ parameters of our local periodic signal model using the SWDFT. Our estimation procedure speeds up computation by using a trigonometric identity that linearizes estimation of $2$ of the $5$ parameters. We conclude with a very small Monte Carlo simulation study of our estimation procedure under different levels of noise. 
\end{abstract}

\section{Introduction}
\label{sec:intro}
Time series methods are typically partitioned into either the time or frequency domain. But there are many cases, such as non-stationary time-series, where the frequency profile changes over time. For these time-series, analysts desire methods that combine both the time {\bf and} frequency domains. Fortunately, scientists and engineers since \cite{gabor1946theory} have invented many ``time-frequency'' methods. This paper focuses on one: The Sliding Window Discrete Fourier Transform (SWDFT). The SWDFT is also known as the ``short-time Fourier transform'', ``windowed Fourier transform'', among other names. 

Our goal is introducing the SWDFT as a tool in the time-series toolbox. We want to show that the SWDFT is useful for local-in-time periodic signals, such as the idealized example shown in Figure \ref{fig:local-periodic}. Just as large discrete Fourier transform (DFT) coefficients provide evidence for global periodic components (\cite{fisher1929tests}), large SWDFT coefficients provide evidence of local periodic components. 

\begin{figure}[ht]
	\centering
	\includegraphics[width = 16cm]{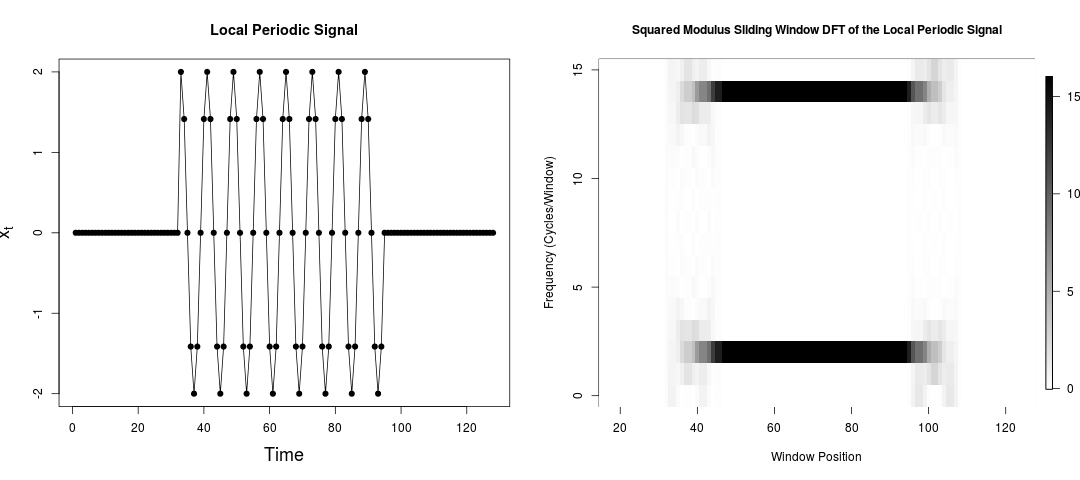} 
	\caption{Left: A local periodic signal: $x_t = 2 \cos(\frac{\pi t}{4}); 31 \leq t \leq 94$. Right: The squared modulus SWDFT coefficients of $x_t$ for a window of length $16$.}
	\label{fig:local-periodic}
\end{figure}

The rest of the paper is organized as follows. Section \ref{sec:swft} defines the SWDFT and some of its key properties. Section \ref{sec:onesignal} introduces a $5$-parameter model for one noiseless local periodic signal, then studies the SWDFT under different parameterizations of this model. Section \ref{sec:rsignals} briefly extends the results from Section \ref{sec:onesignal} to $R > 1$ local periodic signals. Section \ref{sec:estimation} turns to estimation, and we propose a simple estimation procedure for the $5$-parameters of our local periodic signal model using the SWDFT. Section \ref{sec:estimation} concludes with a brief Monte Carlo simulation study on the accuracy of our estimation procedure under different levels of noise. 

\section{The Sliding Window Discrete Fourier Transform (SWDFT)}
\label{sec:swft}
The Sliding Window Discrete Fourier Transform (SWDFT) computes a time-frequency representation of a signal. Simply described, the SWDFT takes sequential discrete Fourier transforms of a signal multiplied by a rectangular sliding window function, where the window function is only nonzero for a short amount of time. This section defines the SWDFT and some of its properties. 

Since the SWDFT is a sequence of DFT's, we first define the DFT. Let ${\bf x} = [x_0, x_1, \ldots, x_{N - 1}]$ be a length $N$ real-valued signal. The DFT of {\bf x} is:

\begin{eqnarray}
	\label{eq:dft}
	a_k &=& \frac{1}{\sqrt{N}} \sum_{j = 0}^{N - 1} x_j \omega_N^{-jk} \nonumber \\
	&& k = 0, 1, \ldots, N - 1 
\end{eqnarray}

\noindent where $\omega_N = e^{\frac{i 2 \pi}{N}} = \cos(\frac{2 \pi}{N}) + i \sin(\frac{2 \pi}{N})$. 

Next, we define the rectangular sliding window function using the following notation for an indicator function:

\[
  \mathbb{1}_{x, y}(j) =
  \begin{cases}
    1 & x \leq j \leq y \\
    0 & \text{ otherwise} \\
  \end{cases}
\]

With this notation, $\mathbb{1}_{p - n + 1, p}(j)$ is the rectangular sliding window function for a length $n$ window at position $p$ in the variable $j$. By applying this window-function to the input signal {\bf x}, the SWDFT is:
	
\begin{eqnarray}
	\label{eq:swft}
	a_{k, p} &=& \frac{1}{\sqrt{n}} \sum_{j = 0}^{n - 1} x_{p - n + 1 + j} \omega_n^{-jk} \nonumber \\
	&& k = 0, 1, \ldots, n - 1 \nonumber \\
	&& p = n - 1, n, \ldots , N - 1 
\end{eqnarray}

\cite{richardson2017algorithm} give an $O(N n)$ algorithm to compute Equation \ref{eq:swft}. In comparison, computing Equation \ref{eq:swft} directly takes $O(N n^2)$ operations, and taking a Fast Fourier Transform (FFT) in each window position takes $O(N n \log(n))$. 

The SWDFT outputs an $n \times P$ array, where $n$ is the number of frequencies and $P = N - n + 1$ is the number of window positions. A useful way to think about the SWDFT output is as a multivariate time-series, where each time-series corresponds to DFT coefficients for a Fourier frequency $(\frac{2 \pi k}{n})$ for $k = 0, \ldots, n - 1$ at each window position (Section 4.1 of \cite{bloomfield2004fourier} gives a detailed description of the Fourier frequencies). For example, $a_{k, .}$ is the time-series for frequency $\frac{2 \pi k}{n}$:

\begin{eqnarray}
	\label{eq:timeseriesk}
	a_{k, .} &=& [a_{k, n - 1}, a_{k, n}, \ldots, a_{k, N - 1}] 
\end{eqnarray} 

\noindent To be concise, we sometimes refer to $a_{k, .}$ as the ``frequency $k$ time-series'', although the frequency is actually $\frac{2 \pi k}{n}$. 

The SWDFT coefficients $a_{k, p}$ are complex numbers, meaning there are different ways to view their output. Some common complex-number outputs are:

\begin{itemize}
	\item $a_{k, p}$: Complex Number
	\item $Re(a_{k, p})$: Real part 
	\item $Im(a_{k, p})$: Imaginary part 	
	\item $|a_{k, p}|^2 = Re(a_{k, p})^2 + Im(a_{k, p})^2$: The squared modulus. 
	\item $\arg(a_{k, p})$: Argument (or, Phase)
\end{itemize}

For local-in-time periodic signals, we focus on the squared modulus of SWDFT coefficients $|a_{k, p}|^2$, since the amplitude is relatively large when the window is on a periodic part of the signal (see Figure \ref{fig:local-periodic}). Following signal processing conventions, we sometimes refer to the squared modulus SWDFT coefficients as the ``energy'' at frequency $k$ and window position $p$. 

\section{The SWDFT for One Local Signal}
\label{sec:onesignal}
Since the SWDFT is useful for local-in-time periodic signals, this section proposes a parametric-model for local-in-time periodic signals, and studies the SWDFT of this model. This study illustrates several key concepts, both mathematically and graphically, for analyzing periodic signals using the SWDFT. We structure this section around three key concepts: Aliasing, Leakage, and Ringing. We selected these three because, in our minds, understanding them greatly enhances understanding of the SWDFT. 

\subsection{Defining a Local Signal}
\label{sec:onelocaldef}
We define a local-in-time periodic signal by applying the indicator function for an interval to a periodic signal. Let $g_t$ be a length $N$ periodic signal:

\begin{eqnarray}
	\label{eq:globalsignal}
	g_t &=& A \cos(\frac{2 \pi F t}{N} + \phi)  \nonumber \\
	&& t = 0, 1, \ldots N - 1
\end{eqnarray}

\noindent Henceforth, when we say ``local periodic signal'', we are referring to a local-in-time periodic signal. Let $x_t$ be a local periodic signal:

\begin{eqnarray}
	\label{eq:localsignal}
	x_t &=& g_t \cdot \mathbb{1}_{S, S + L - 1}(t)
\end{eqnarray}

\noindent Local periodic signals $(x_t)$ have five parameters:

\begin{itemize}
	\item $S$: Start in time of local signal. $S \in \{0, 1, \ldots, N - 2\}$. Integer
	\item $L$: Length in time of local signal. $L \in \{1, 2, \ldots, N - S\}$. Integer
	\item $A$: Amplitude. $A \in [0, \infty]$. Real Number 
	\item $F$: Frequency, number of complete cycles in length $N$ signal. $F \in [0, \infty]$. Real Number
	\item $\phi$: Phase. $\phi \in [0, 2 \pi]$. Real Number 
\end{itemize}

We re-write the SWDFT of $x_t$ in a form that helps explain its behavior. Recall Euler's identity:

\begin{eqnarray}
	\label{eq:euler}
	\cos(x) &=& \frac{1}{2} (e^{ix} + e^{-ix}) \nonumber 
\end{eqnarray}

\noindent To be concise, let $\hat{p} = p -n + 1$. Then substitute $x_t$ into the SWDFT (Equation \ref{eq:swft}):

\begin{eqnarray}
	\label{eq:local-signal-swft}
	a_{k, p} &=& \frac{1}{\sqrt{n}} \sum_{j = 0}^{n - 1} x_{\hat{p} + j} \omega_n^{-jk} \nonumber \\	
	&=& \frac{1}{\sqrt{n}} \sum_{j = 0}^{n - 1} A \cos(\frac{2 \pi (\hat{p} + j) F}{N} + \phi) \mathbb{1}_{S, S + L - 1}(\hat{p} + j) \omega_n^{-jk} \nonumber \\
	&=& \frac{A}{2 \sqrt{n}} \sum_{j = 0}^{n - 1} [e^{i(\frac{2 \pi (\hat{p} + j) F}{N} + \phi)} + e^{-i(\frac{2 \pi (\hat{p} + j) F}{N} + \phi)}] \mathbb{1}_{S, S + L - 1}(\hat{p} + j) \omega_n^{-jk} \nonumber \\
	&=& \frac{A}{2 \sqrt{n}}[e^{i \phi} \omega_n^{f \hat{p}} \sum_{j = 0}^{n - 1} \omega_n^{-j(k - f)} \mathbb{1}_{S, S + L - 1}(\hat{p} + j) + e^{-i \phi} \omega_n^{-f \hat{p}} \sum_{j = 0}^{n - 1} \omega_n^{-j(k + f)} \mathbb{1}_{S, S + L - 1}(\hat{p} + j)] \nonumber \\
\end{eqnarray}

\noindent where $f = \frac{nF}{N}$: the number of cycles in a length $n$ window. We return to this expression several times throughout the paper. 

\subsection{Aliasing}
\label{sec:aliasing}
In the discrete Fourier setting, since we only have $N$ data-points, it is important to understand the range of frequencies we can detect. Fortunately, the detectable frequencies are precisely described by the concept of ``aliasing''. Aliasing means that one frequency serves as an ``alias'' for another, and the two frequencies are indistinguishable with discrete data. The largest detectable frequency is known as the {\bf Nyquist frequency}, defined as $\frac{1}{2 \delta}$, where $\delta$ is the size of the sampling interval (we assume $\delta = 1$ henceforth). The Nyquist frequency is the largest detectable frequency because any frequency larger than $\frac{1}{2}$ can be {\it folded} into the range $[0, \frac{1}{2}]$. To see this, recall our periodic signal:

\begin{eqnarray}
	g_t &=& A \cos(\frac{2 \pi t F}{N} + \phi)
\end{eqnarray}

\noindent assume $A = 1$ and $\phi = 0$ for simplicity, and let $f = \frac{F}{N}$ for clarity. Next, let $Q$ be an integer such that $Q - f = f' \in [0, \frac{1}{2}]$. Then we have:

\begin{eqnarray}
	\cos(2 \pi t f) &=& \cos(2 \pi t (Q - f')) \nonumber  \\
	&=& \cos(2 \pi t Q - 2 \pi t f') \nonumber \\
	&=& \cos(-2 \pi t f') \nonumber \\
	&=& \cos(2 \pi t f') \nonumber 
\end{eqnarray}
	
The second-to-third line follows because $Qt$ is an integer and cosine has a period of $2 \pi$, and the third-to-fourth line follows because $\cos(-x) = \cos(x)$. The same derivation shows that $\sin(2 \pi f t) = -\sin(2 \pi f' t)$ (Section 2.5 of \cite{bloomfield2004fourier}). 

This derivation shows that frequencies $f$ and $f'$ are aliases, and more generally, any frequency $f$ can be {\it folded} into the range $[0, \frac{1}{2}]$. Frequencies in this range are commonly called the {\bf principal aliases}. It is typical to restrict SWDFT analysis to frequencies $\frac{k}{n} \in [0, \frac{1}{2}]$, meaning that we only consider the coefficients $0 \leq k \leq \frac{n}{2}$.

In sum, aliasing is important because it provides a straightforward description of detectable frequencies. For example, if you are designing an experiment that looks for oscillations at a specific frequency (say, $60$ Hz), aliasing determines how fast you need to sample. 

\subsection{Leakage}
\label{sec:globalperiodic}
The next key concept is {\bf leakage}. Leakage occurs when the local periodic signal's components ($\frac{2 \pi f}{n}$) are {\bf not} at Fourier frequencies. Alternatively, we can say that leakage occurs when the number of cycles in a length $n$ window ($f = \frac{nF}{N}$) is not an integer. Leakage means that energy ``leaks'' into all SWDFT coefficients. With no leakage, the only nonzero SWDFT coefficients occur at the Fourier frequency matching the true frequency component. Figure \ref{fig:leakage-noleakage} compares two signals with and without leakage. 

\begin{figure}[ht]
	\centering
	\includegraphics[width = 12cm]{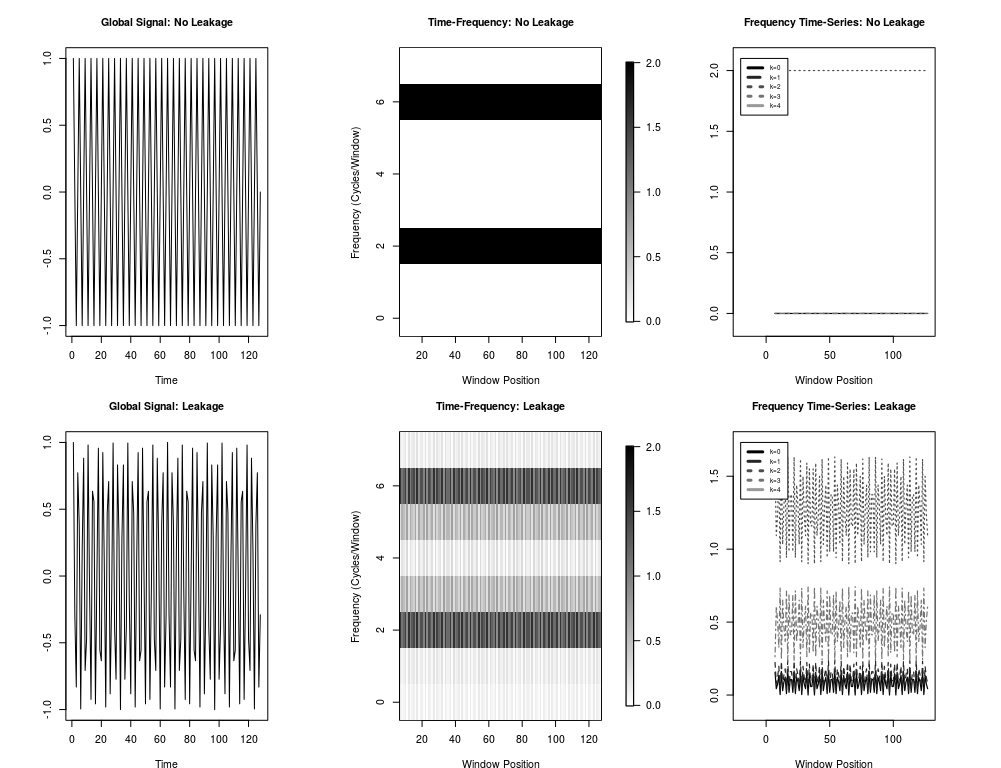} 
	\caption{Two global periodic signals, with (bottom panels), and without (top panels) leakage. The top left panel is the signal $\cos(\frac{2 \pi 32 t}{128})$, resulting in $f = 2$ cycles per window. Since there is an integer number of cycles, the SWDFT has {\bf no leakage}, and all the energy is concentrated at frequency $f = 2$ and its alias $f = 6$. The top middle panel is the standard time-frequency plot as in Figure \ref{fig:local-periodic}. The top right panel shows each unique squared-modulus frequency time-series, in this case $|a_{0, .}|^2, |a_{1, .}|^2$, \ldots, $|a_{4, .}|^2$. The bottom three panels show the same plots, except now they correspond to the signal $\cos(\frac{2 \pi 38 t}{128})$, which results in $f = 2.375$ cycles per window. Since there is {\bf leakage}, there is nonzero energy at all frequency time-series, and the largest share goes to $k = 2$: the Fourier frequency closest to the true frequency ($f = 2.375$).}
	\label{fig:leakage-noleakage}	
\end{figure}

To understand leakage mathematically, we slightly adapt the Dirichlet kernel $D_n$ from Section 2.2 of \cite{bloomfield2004fourier}. Let $\mathbb{N}$ be the natural numbers, then the Dirichlet kernel is:

\[
  \label{eq:dirichlet-kernel}
  D_n(x) =
  \begin{cases}
  	n & x = 2 \pi m , m \in 2 \mathbb{N} \\
  	-n & x = 2 \pi m , m \in 2 \mathbb{N} + 1 \\
  	0 & x = \frac{2 \pi m}{n}, m \in \mathbb{N} \\
  	\frac{ \sin(\frac{nx}{2}) }{ \sin(\frac{n}{2}) } & \text{otherwise} 
  \end{cases}
\]

\noindent where $n$ is an integer (in our case, the SWDFT window size). Assume $S = 0$ and $L = N$ for simplicity. Then using the following identity for sums of complex exponentials (derived in Chapter 2 of \cite{richardson2018thesis}):

\begin{eqnarray}
	\label{eq:general-complex-sum}
	\sum_{j = a}^{b} e^{i j x}  &=& e^{ix \frac{(a + b)}{2}} D_{b - a + 1}(x)
\end{eqnarray}

\noindent We can re-write Equation \ref{eq:local-signal-swft} using the Dirichlet kernel and Equation \ref{eq:general-complex-sum}:

\begin{eqnarray}
	\label{eq:signaldir}
	a_{k, p} &=& \frac{A}{\sqrt{2n}}[e^{i \phi} \omega_N^{F \hat{p}} e^{-i(\frac{n - 1}{n})(k - f)} D_n(\frac{2 \pi (k - f)}{n}) + e^{-i \phi} \omega_N^{-F \hat{p}} e^{-i(\frac{n - 1}{n})(k + f)} D_n(\frac{2 \pi (k + f)}{n})] \nonumber \\
\end{eqnarray}

The key point is that since $k$ is always an integer, when $f$ is an integer, $k \pm f$ is an integer. And when $k \pm f$ is an integer, $D_n(\frac{2 \pi (k \pm f)}{n}) = 0$, unless $k = f$. And finally, when $D_n(\frac{2 \pi (k \pm f)}{n}) = 0$, Equation \ref{eq:signaldir} equals 0. This explains why there is no leakage ($a_{k, p} = 0$) when $k \neq f$ and $f$ is an integer. When $k = f$, Equation \ref{eq:signaldir} simplifies to:

\begin{eqnarray}
	\noindent 
	a_{f, p} &=& A \sqrt{\frac{n}{2}} e^{i(\frac{2 \pi f \hat{p}}{n} + \phi)}
\end{eqnarray}

\noindent And the real, imaginary, and squared modulus parts are:

\begin{eqnarray}
	Re(a_{f, p}) &=& A \sqrt{\frac{n}{2}} \cos(\frac{2 \pi f \hat{p}}{n} + \phi) \nonumber \\
	Im(a_{f, p}) &=& A \sqrt{\frac{n}{2}} \sin(\frac{2 \pi f \hat{p}}{n} + \phi) \nonumber \\
	|a_{f, p}|^2 &=& \frac{A^2 n}{2} \nonumber \\
\end{eqnarray}

Notably, the squared modulus $|a_{f, p}|^2$ is large and constant, and the real and imaginary parts are smaller and oscillate. This explains why the squared modulus SWDFT coefficients are better at identifying local periodic signals than either the real or imaginary parts. 

The only exception is when $k = f = \frac{n}{2}$, and $n$ is even. In this case, the second term in Equation \ref{eq:signaldir} does not vanish, and we have:

\begin{eqnarray}
	a_{\frac{n}{2}, p} &=& A \sqrt{\frac{n}{2}} [e^{i(\frac{2 \pi f \hat{p}}{n} + \phi)} - e^{i(\frac{-2 \pi f \hat{p}}{n} + \phi + \frac{n - 1}{n})}]
\end{eqnarray}

Figure \ref{fig:dirichlet} shows what happens to the Dirichlet kernel when $k = f = \frac{n}{2}$. Since the Dirichlet kernel is periodic, both $D_n(0)$ and $D_n(n)$ equal $n$. 

That explains why there is no leakage when $f$ is an integer. When $f$ is {\it not} an integer, the bottom panel of Figure \ref{fig:leakage-noleakage} shows a signal and its corresponding SWDFT. The key observation is that while all SWDFT coefficients are nonzero, the coefficients {\bf closest} to the true frequency are the largest. We further illustrate this important fact in Figure \ref{fig:frequency-evolution}. Figure \ref{fig:frequency-evolution} shows the squared modulus of frequency time-series $2$ and $3$ ($|a_{2, .}|^2$) and $|a_{3, .}|^2$) for ten different periodic signals with frequencies $f = 2, 2.1, \ldots, 2.9$. The key point is that as the true frequency gets {\bf closer} to a particular Fourier frequency, the size of the squared modulus coefficients at that frequency get larger. 

\begin{figure}[ht]
	\centering
	\includegraphics[width = 16cm]{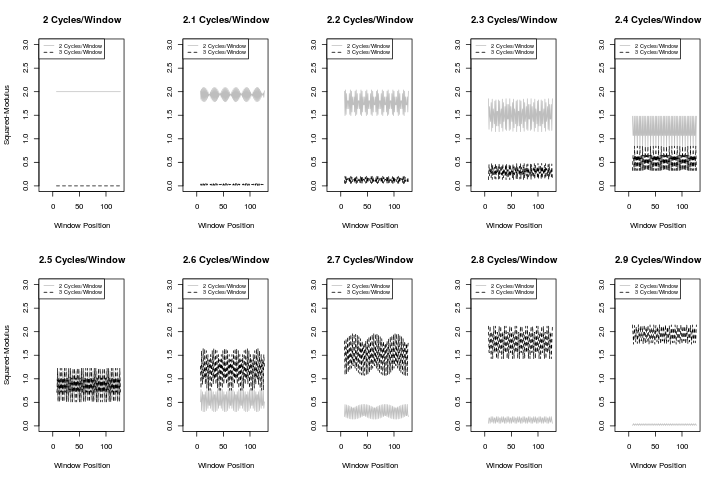} 
	\caption{Time-series of squared modulus SWDFT coefficients corresponding to the Fourier frequencies with two cycles per window $a_{2, .} = [a_{2, n - 1}, a_{2, n}, \ldots, a_{2, N - 1}]$ and three cycles per window $a_{3, .} = [a_{3, n - 1}, a_{3, n}, \ldots, a_{3, N - 1}]$, where the true number of cycles per window ranges from $f = 2, 2.1, \ldots, 2.9$. This illustrates that Fourier frequencies {\bf closest} to the true frequency have the largest squared modulus coefficients. For example, when $f = 2$, coefficients $a_{2, .}$ get 100$\%$ of the energy, and when $f = 2.5$, coefficients $a_{2, .}$ and $a_{3, .}$ split the energy (almost) equally.}
	\label{fig:frequency-evolution}
\end{figure}

While Figure \ref{fig:frequency-evolution} gives a graphical explanation, the Dirichlet kernel provides a mathematical explanation of why Fourier frequencies close to the true frequency are larger. Figure \ref{fig:dirichlet} shows both the Dirichlet kernel (top panel) and the Dirichlet Weight (bottom panel), where the Dirichlet Weight is the Dirichlet Kernel multiplied by the complex constant that appears in Equation \ref{eq:signaldir}. Both quantities {\bf peak} when $k \pm f = 0$. But, since both $D_n(\frac{2 \pi (k \pm f)}{n})$ and $DW_n(\frac{2 \pi (k \pm f)}{n})$ are continuous, their values remain large when $k \pm f$ is close to 0. This explains why even if $f$ is not an integer, Fourier frequencies {\bf closest} to $f$ have the largest SWDFT coefficients. 

In sum, leakage occurs when the true frequency of a periodic signal is not a Fourier frequency. With leakage, all SWDFT coefficients are nonzero, but the Fourier frequency {\bf closest} to the true frequency will have the largest squared modulus SWDFT coefficients. 

\begin{figure}[ht]
	\centering
	\includegraphics[width = 12cm, height=10cm]{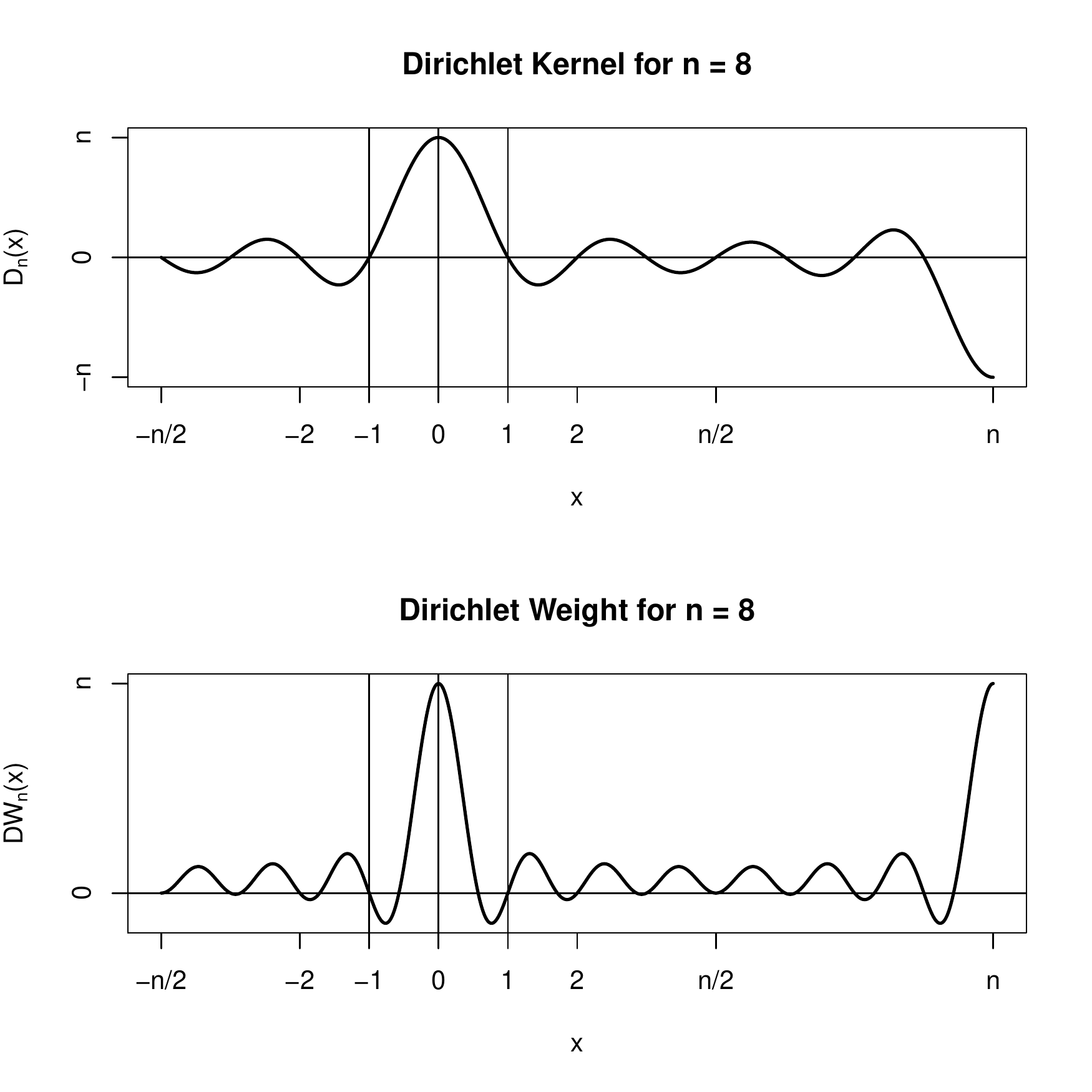} 
	\caption{The top panel shows the Dirichlet kernel $D_n(x)$ for $n = 8$, and the bottom panel shows the Dirichlet kernel multiplied by a complex constant, called the Dirichlet weight ($DW_n(x)$), for $n = 8$. The two key features are that both $D_n(x)$ and $DW_n(x)$ peak at $n$ when $x = 0$, and equal $0$ for all integer values of $x$. The shape of a Dirichlet kernel is remarkably similar to the $sinc(x)$ function (\cite{weisstein2002sinc}), which is sometimes used to approximate it.}
	\label{fig:dirichlet}	
\end{figure}

\subsection{Ringing and Trapezoids}
\label{sec:ringing-localsignals}
The previous sections focused on global periodic signals ($S = 0, L = N$) in order to simplify our explanation of Aliasing and Leakage. But the purpose of the SWDFT is analyzing {\bf local} periodic signals. Otherwise, we would just use the DFT. Therefore, this section assumes that $S > 0$ and $L < N$, which implies that $x_t$ is a local periodic signal. 

Following the definition in Section \ref{sec:onelocaldef}, local periodic signals oscillate for an interval of time, and are zero-valued for the rest of the signal. We refer to this as the ``oscillating part'' and ``zero part'' of the local periodic signal, respectively. Since the SWDFT takes a DFT in each window position, there are inevitably {\it some} window positions that cover both the oscillating and zero parts of the signal. The window positions that cover both parts lead to a phenomena called ``ringing''. 

We refer to ringing as the undesired detection of frequencies due to sharp transitions in the input signal. Another common definition of ringing is when the DFT is applied to a discontinuous functions (\cite{okamura2011short}). Ringing is closely related to the ``Gibbs Phenomena'', which is the fact that a Fourier series expansion can not exactly recover a function with a jump discontinuity. For local periodic signals, the ``sharp transition'' occurs when the signal changes from the zero part to the oscillating part, or vice versa. With our definition, ringing still occurs even if the local periodic signal is continuous, since the frequency profile differs on the zero and oscillating parts of the signal.  

The step-function provides a simple demonstration of ringing in the SWDFT due to a sharp transition in the input signal. Following \cite{okamura2011short}, define the step-function as:

\[
  s_t =
  \begin{cases}
    1 & t \geq d \\
    0 & \text{ otherwise} \\
  \end{cases}
\]

\noindent The SWDFT of $s_t$ for length $n$ windows is:

\begin{eqnarray}
	\label{eq:swdftstep}
	a_{k, p} &=& \frac{1}{\sqrt{n}} \sum_{j=0}^{n - 1} s_{p - n + 1 + j} \omega_n^{-jk}
\end{eqnarray}

The SWDFT coefficients of the step-function (Equation \ref{eq:swdftstep}) depend 
on whether the window position occurs before, during, or after the ``step'' in $s_t$ from 0 to 1 at time $d$. If the window position occurs before the jump ($p < d$), then all SWDFT coefficients are zero. When the window position is after the jump ($p > d + n - 1$), the coefficients are:

\[
  a_{k, p} =
  \begin{cases}
    \sqrt{n} & k = 0 \\
    0 & \text{otherwise} \\
  \end{cases}
\]

Ringing occurs when the window position contains the ``step'': $d \leq p \leq d + n - 1$. In this case, the coefficients are:

\[
	a_{k, p} = 
	\begin{cases}
		\frac{p - d + 1}{\sqrt{n}} & k = 0\\
		\frac{\omega_n^{-k(d - p + n - 1)}}{\sqrt{n}} \cdot \frac{(1 - \omega_n^{-k(p - d + 1)})}{(1 - \omega_n^{-k})}  & \text{otherwise}
	\end{cases}
\]

There is ringing here because all coefficients are nonzero, even though there is no periodic component in the input signal. The sharp transition from $0$ to $1$ at time $d$ causes the ringing. 

The SWDFT of local periodic signals has ringing for the same reason the step-function has ringing: some window positions are on both the zero and oscillating parts of the signal. For these window positions, the SWDFT coefficients are nonzero at all frequencies, even when the signal has no leakage. We now give a precise description of ringing for local periodic signals. 

\begin{figure}[ht]
	\centering
	\includegraphics[width = 14cm]{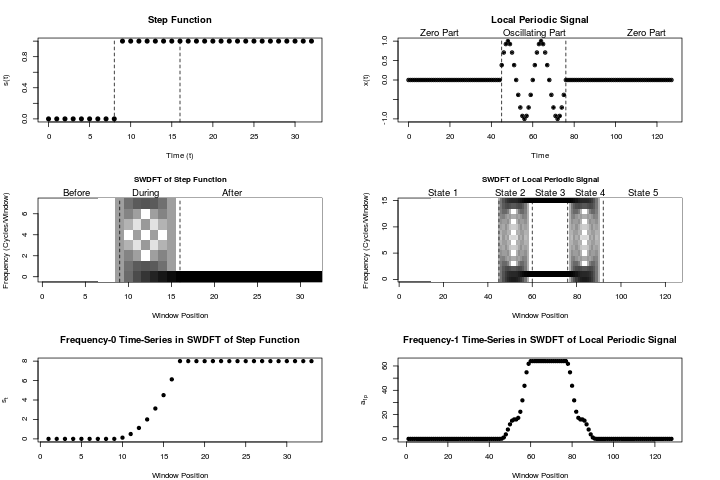} 
	\caption{Demonstration of the ringing phenomenon in the SWDFT for both a step-function (left panels) and a local periodic signal (right panels). The SWDFT of the step-function is broken into three states, depending on the window position: Before, During, and After the step from $0$ to $1$ at time $d$. Ringing occurs in the ``During'' state, that is, when the window position covers the step. The bottom-left panel shows the frequency-$0$ time-series ($|a_{0, }|^2$) for all window positions in the SWDFT of the step-function. For local periodic signals, ringing occurs in States 2 and 4, when the window position covers both the zero and oscillating parts of the input signal. The bottom-right panel shows the frequency-$1$ time-series $|a_{k, .}|^2$, which displays the trapezoid-like shape of Figure \ref{fig:canonical}. We use the log of the squared modulus coefficients in both SWDFT plots $\log(|a_{k, p}|^2)$, to emphasize the fact that the SWDFT coefficients are nonzero.}
	\label{fig:stepswdft}	
\end{figure}

For the sake of exposition, we assume that the local signal starts sufficiently far from the beginning of the time-series ($S > 2n$) and that the signal length is larger than the window size $L > n$, although these assumptions are not critical. Whereas the window position could be in three different states (before, during, after) for the step-function, the window position can be in {\bf five different states} for a local periodic signal:

\begin{itemize}
	\item State 1: Before the signal enters the window; 
	\item State 2: The first $q$ window points are on the local signal;
	\item State 3: The window is only on the local signal;
	\item State 4: The last $q$ points are on the local signal;
	\item State 5: After the signal passed the window;
\end{itemize}

States 1 and 5 are zero for all frequencies, and state 3 is identical to the SWDFT for global periodic signals (see Section \ref{sec:globalperiodic}). Ringing occurs in states 2 and 4, when the window position is on both the zero part and oscillating part of the local periodic signal. The exact SWDFT coefficients in each state are:

\begin{eqnarray}
\footnotesize	
  a_{k, p} =
  \begin{cases}
 	0 & \text{State 1} \\
  	\frac{A}{\sqrt{2n}}[e^{i \phi} \omega_n^{f \hat{p}} e^{-i(\frac{-\pi(2n - q - 1)(k - l)}{n})} D_q(\frac{2 \pi (k - f)}{n}) + e^{-i \phi} \omega_n^{-f\hat{p}} e^{-i(\frac{-\pi(2n - q - 1)(k + f)}{n})} D_q(\frac{2 \pi (k + f)}{n})] & \text{State 2} \\
  	\frac{A}{\sqrt{2n}}[e^{i \phi} \omega_n^{f \hat{p}} e^{-i(\frac{n - 1}{n})(k - f)} D_n(\frac{2 \pi (k - f)}{n}) + e^{-i \phi} \omega_n^{-f\hat{p}} e^{-i(\frac{n - 1}{n})(k + f)} D_n(\frac{2 \pi (k + f)}{n})] & \text{State 3} \\
  	\frac{A}{\sqrt{2n}}[e^{i \phi} \omega_n^{f \hat{p}} e^{-i(\frac{-\pi(q - 1)(k - f)}{n})} D_q(\frac{2 \pi (k - f)}{n}) + e^{-i \phi} \omega_n^{-f\hat{p}} e^{-i(\frac{-\pi(q - 1)(k + f)}{n})} D_q(\frac{2 \pi (k + f)}{n})] & \text{State 4} \\
  	0 & \text{State 5} 
  \end{cases} 
    \label{eq:swft-coef-local-signal} 
\end{eqnarray}

\noindent where $k = 0, 1, \ldots, n - 1$ and $p = n - 1, n, \ldots, N - 1$. 

States $2$ and $4$ of Equation \ref{eq:swft-coef-local-signal} give mathematical explanations for ringing in local periodic signals. In these states, the Dirichlet kernel changes from $D_n(x)$ in global periodic signals (and state 3) to $D_q(x)$ in states 2 and 4. And since since $D_n(x)$ peaks at $n$ (see Section \ref{sec:globalperiodic}), this implies that $D_q(x)$ peaks at $q \leq n$, which explains why the SWDFT coefficient in states 2 and 4 are {\bf smaller} than state 3. Another important point regarding states 2 and 4 is that when $k = f$ (e.g. there is no leakage), the SWDFT coefficients are still nonzero. To see why, consider the Dirichlet kernel in states 2 and 4:

\vspace{-1em}
\begin{eqnarray}
	D_q(\frac{2 \pi x}{n}) &=& \frac{\sin(\frac{q \pi x}{n})}{\sin(\frac{\pi x}{n})}\nonumber 
\end{eqnarray}
\vspace{-1em}

\noindent which means we need $\frac{qx}{n}$ to be an integer for these coefficients to vanish. The right panels of Figure \ref{fig:stepswdft} illustrate the SWDFT coefficients in the five different states.

\begin{figure}[ht]
	\centering
	\includegraphics[width = 12cm]{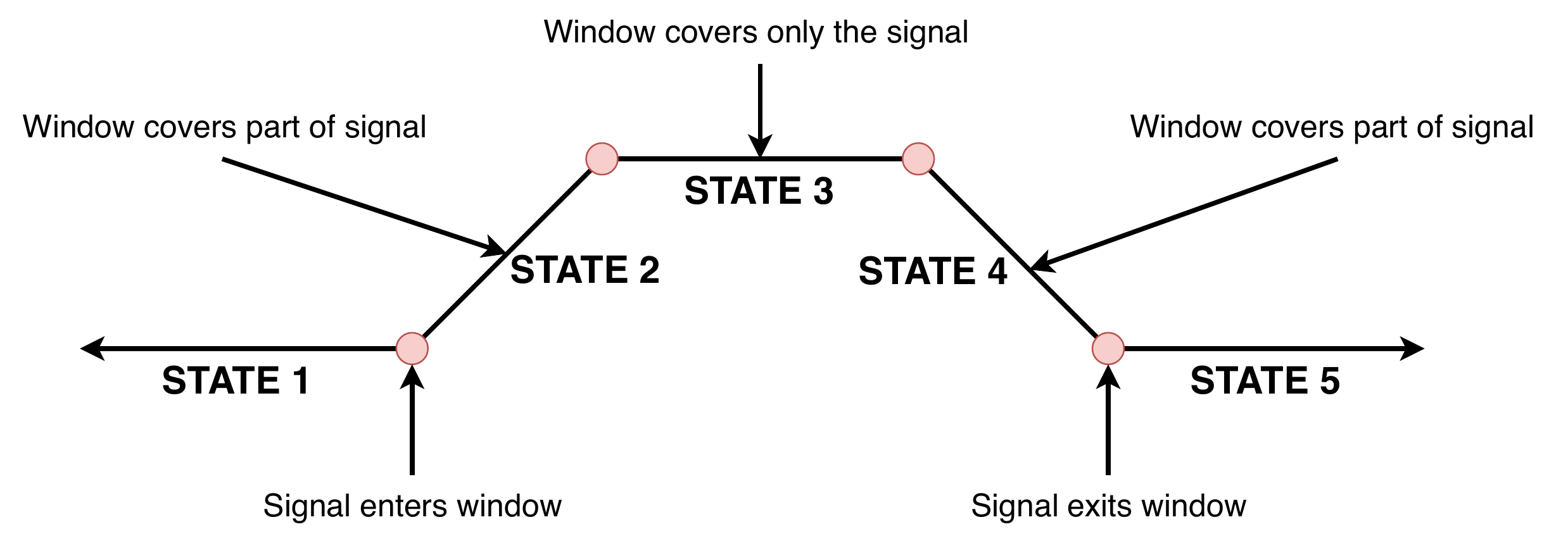} 
	\caption{The shape of a squared modulus frequency time-series $|a_{k, .}|^2$ corresponding to the true frequency of a local periodic signal.}
	\label{fig:canonical}	
\end{figure}

We conclude with a final point regarding the shape of the frequency time-series corresponding to the true frequency component of a local periodic signal. Figure \ref{fig:canonical} caricatures this shape, which looks more or less like a trapezoid. It is a caricature, however, because in practice we do not get straight lines in states 2 and 4 (the lines oscillate just like for global signals, but also increase/decrease as points enter/exit the window, see bottom-right panel of Figure \ref{fig:stepswdft}). 

In sum, local periodic signals have two states: the zero part and the oscillating part. Ringing, defined as the undesired detection of frequencies due to sharp transitions in the input signal, occurs at window positions that are on both the zero and oscillating parts of the signal simultaneously. 

\section{R Local Signals}
\label{sec:rsignals}
Now that we understand the SWDFT for one local periodic signal, we briefly generalize to $R$ local periodic signals. They key idea is that the SWDFT of $R$ local signals is the sum of the SWDFT for each local signal:

\begin{theorem}[$R$ Local Signals]
	\label{thm:rlocalsignals}
	Define $R$ length $N$ local periodic signals:

	\begin{eqnarray}
		{\bf x}_r &=& [x_{0, r}, x_{1, r}, \ldots, x_{N - 1, r}] \nonumber \\
		r &=& 1, 2, \ldots, R \nonumber \\
		t &=& 0, 1, \ldots, N - 1 \nonumber 
	\end{eqnarray}

	\noindent where

	\begin{eqnarray}
		x_{t, r} &=& A_r \cos(\frac{2 \pi t F_r}{N} + \phi_r) \cdot \mathbb{1}_{S_r, S_r + L_r - 1}(t) \nonumber 
	\end{eqnarray}

	\noindent Let $y_t$ be the sum of $R$ local period signals:

	\begin{eqnarray}
		y_t &=& \sum_{r = 1}^{R} x_{t, r}
	\end{eqnarray}

	\noindent Let $b_{k, p, r}$ be the SWDFT for $x_{t, r}$ for window size $n$. Then the SWDFT of $y_t$ is:

	\begin{eqnarray}
		a_{k, p} &=& \sum_{r = 1}^{R} b_{k, p, r} \nonumber \\
		k &=& 0, 1, \ldots n - 1 \nonumber \\
		p &=& n - 1, n, \ldots, N - 1 \nonumber 
	\end{eqnarray}

\end{theorem}

\begin{proof}
	We prove Theorem \ref{thm:rlocalsignals} by substituting $x_{t, r}$ into Equation \ref{eq:swft} (the SWDFT definition):

	\begin{eqnarray}
		a_{k, p} &=& \frac{1}{\sqrt{n}} \sum_{j = 0}^{n - 1} y_{p - n + 1 + j} \omega_n^{-jk} \nonumber \\
		&=& \frac{1}{\sqrt{n}} \sum_{j = 0}^{n - 1} (\sum_{r = 1}^{R} x_{p - n + 1 + j, r}) \omega_n^{-jk} \nonumber \\
		&=& \sum_{r = 1}^{R} \frac{1}{\sqrt{n}} \sum_{j = 0}^{n - 1} x_{p - n + 1 + j, r} \omega_n^{-jk} \nonumber \\
		&=& \sum_{r = 1}^{R} b_{k, p, r}
	\end{eqnarray}

\end{proof}

Figure \ref{fig:two-signals} gives two simple examples of Theorem \ref{thm:rlocalsignals}: when the signals overlap, and when they do not. The SWDFT of the two signals is clearly additive. 

\begin{figure}[ht]
	\centering
	\includegraphics[width = 16cm]{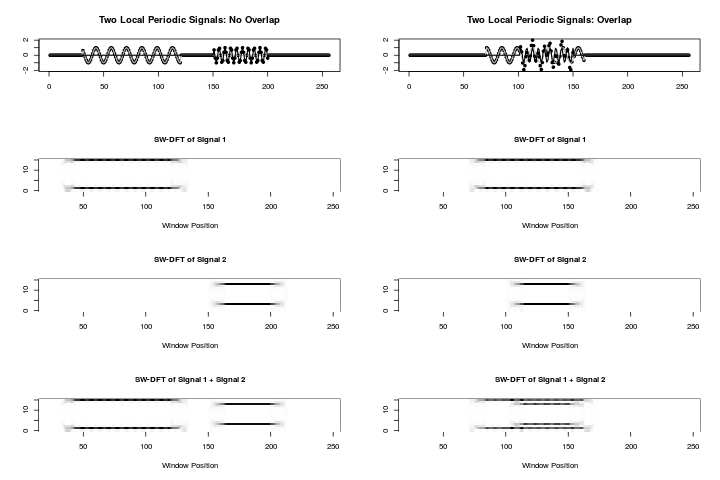} 
	\caption{Comparison of two signals with two periodic components each. The four panels on the left correspond to a signal where the two periodic components {\bf do not} overlap in time. The top-left plot shows the two signals. The gray line is the first periodic component, the black line is the second periodic component, and the black dots are the sum of both periodic components. The second and third plots on the left-side show SWDFT coefficients for each periodic component separately, and the bottom-left plot shows the SWDFT coefficients for the sum of the periodic components. The four right-side plots are the same as the left-side plots, except now the two periodic components overlap.}
	\label{fig:two-signals}	
\end{figure}

\section{Estimation}
\label{sec:estimation}
The previous sections studied SWDFT output under known, local periodic input signals. This section reverses this, and studies estimation of local periodic signals from the SWDFT. Specifically, this section addresses the following question: assuming our time-series is the sum of $R$ local periodic signals, how do we estimate the parameters? Like Section \ref{sec:onesignal}, we start with the noiseless case, then add noise in the final section. Our model is:

\begin{eqnarray}
	\label{eq:rlocalsignals}
	y_t &=& \sum_{i = 1}^{R} A_r \cos(\frac{2 \pi t F_r}{N} + \phi_r) \mathbb{1}_{S_r, S_r + L_r - 1}(t) \nonumber \\
	&& t = 0, 1, \ldots, N - 1
\end{eqnarray}

\noindent Equation \ref{eq:rlocalsignals} has $5R + 1$ parameters, since each local periodic signal has five parameters: ($A_r, F_r, \phi_r, S_r, L_r$), and we need to estimate $R$. The SWDFT of Equation \ref{eq:rlocalsignals} with length $n$ windows is:

\begin{eqnarray}
	\label{eq:rlocalswft}
	a_{k, p} &=& \frac{1}{\sqrt{n}} \sum_{j = 0}^{n - 1} \sum_{r = 1}^{R} A_r \cos(\frac{2 \pi t F_r}{N} + \phi_r) \mathbb{1}_{S_r, S_r + L_r - 1} \omega_n^{-jk}
\end{eqnarray}

Equation \ref{eq:rlocalswft} is our model for the SWDFT of $R$ local periodic signals. Next, we define the data. Let ${\bf x} = [x_0, x_1, \ldots, x_{N - 1}]$ be a length $N$ signal, and let ${\bf b}$ be the SWDFT of {\bf x} with length $n$ windows:

\begin{eqnarray}
	b_{k, p} &=& \frac{1}{\sqrt{n}} \sum_{j = 0}^{n - 1} x_{p - n + 1 + j} \omega_n^{-jk} \nonumber \\
	&& k = 0, 1, \ldots, n - 1 \nonumber \\
	&& p = n - 1, n, \ldots, N - 1 
\end{eqnarray}	

\noindent We want the least-squares parameter estimates of $a_{k, p}$:

\begin{eqnarray}
	\argmin_{R, S_r, L_r, A_r, F_r, \phi_r; \atop r = 1, \ldots R} && \sum_{k = 0}^{n - 1} \sum_{p = n - 1}^{N - 1} (b_{k, p} - a_{k, p})^2
\end{eqnarray}

\noindent We focus on the case where $R = 1$, although the concepts extend to $R > 1$. When $R > 1$, however, the computational complexity rapidly increases (in proportion to the total number of parameters), which implies that the success of our estimation procedure largely depends on optimization methods and computational implementation. We leave a complete estimation procedure to the future. When $R = 1$, we need to estimate five parameters of $a_{k, p}$:

\begin{eqnarray}
	\argmin_{S, L, F, A, \phi} && \sum_{k = 0}^{n - 1} \sum_{p = n - 1}^{N - 1} (b_{k, p} - a_{k, p})^2
\end{eqnarray}

\noindent Our SWDFT model for $R = 1$ is:

\begin{eqnarray}
	a_{k, p} &=& \frac{1}{\sqrt{n}} \sum_{j = 0}^{n - 1} A \cos(\frac{2 \pi (\hat{p} + j) F}{N} + \phi) \mathbb{1}_{S, S + L - 1}(\hat{p} + j) \omega_n^{-jk}
\end{eqnarray}

\subsection{Linearizing The Model}
\label{sec:linearizing}
The key estimation idea is that if we fix $S, L$, and $F$, we can linearize the model and efficiently estimate $A$ and $\phi$. Based on a trick from Section 2.1 of \cite{bloomfield2004fourier}, re-write the cosine factor using the following trigonometric identity:

\begin{eqnarray}
	A \cos(\frac{2 \pi (\hat{p} + j) F}{N} + \phi) &=& \beta_1 \cos(\frac{2 \pi (\hat{p} + j) F}{N}) + \beta_2 \sin(\frac{2 \pi (\hat{p} + j) F}{N})	
\end{eqnarray}

\noindent where 

\begin{eqnarray}
	\beta_1 &=& A \cos(2 \pi \phi) \nonumber \\
	\beta_2 &=& -A \sin(2 \pi \phi) \nonumber 	
\end{eqnarray}

\noindent And, if we know $\beta_1$ and $\beta_2$, there is a one-to-one mapping to $A$ and $\phi$:

\begin{eqnarray}
	A &=& \sqrt{\beta_1^2 + \beta_2^2} \nonumber \\
	\phi &=& \frac{\arctan(\frac{-\beta_2}{\beta_1})}{2 \pi} \nonumber 	
\end{eqnarray}

\noindent Define the complex-valued constants $C_{1, k, p}$ and $C_{2, k, p}$ as:

\begin{eqnarray}
	C_{1, k, p} &=& \frac{1}{\sqrt{n}} \sum_{j = 0}^{n - 1} \cos(\frac{2 \pi (\hat{p} + j)F}{N}) \mathbb{1}_{S, S + L - 1}(\hat{p} + j) \omega_n^{-jk}  \nonumber \\
	C_{2, k, p} &=& \frac{1}{\sqrt{n}} \sum_{j = 0}^{n - 1} \sin(\frac{2 \pi (\hat{p} + j)F}{N}) \mathbb{1}_{S, S + L - 1}(\hat{p} + j) \omega_n^{-jk} \nonumber 	
\end{eqnarray}

\noindent Then re-write our model as linear:

\begin{eqnarray}
	b_{k, p} &=& \beta_1 C_{1, k, p} + \beta_2 C_{2, k, p}
\end{eqnarray}

\noindent Now our optimization problem is:

\begin{eqnarray}
	\argmin_{S, L, F}( \argmin_{A, \phi} && \sum_{k = 0}^{n - 1} \sum_{p = n - 1}^{N - 1} (b_{k, p} - \beta_1 C_{1, k, p} - \beta_2 C_{2, k, p})^2)
\end{eqnarray}

\noindent And since the inner-most optimization is linear, we can solve it efficiently. One complication is that $b_{k, p}$, $C_{1, k, p}$, and $C_{2, k p}$ are complex numbers. Since the $\beta$'s are real-valued, we can write the model in complex-number notation:

\begin{eqnarray}	
	\label{eq:complexreg}
	[Re(b_{k, p}), Im(b_{k, p})] &=& [\beta_1 Re(C_{1, k, p}), \beta_1 Im(C_{1, k, p})] + [\beta_2 Re(C_{2, k, p}), \beta_2 Im(C_{2, k, p})] \nonumber \\
\end{eqnarray}	 

\noindent Then by the definition of complex addition, we can solve for the $\beta$'s using the real parts of Equation \ref{eq:complexreg}:

\begin{eqnarray}
	Re(b_{k, p}) &=& \beta_1 Re(C_{1, k, p}) + \beta_2 Re(C_{2, k, p}) \nonumber 
\end{eqnarray}

\subsection{Optimizing $F$ and Selecting $k$}
\label{sec:optfk}
Linearizing the model means that if we are given $S$, $L$, and $F$, we can efficiently estimate $A$ and $\phi$. However, we still require estimates for $S$, $L$, and $F$. Since $S$ and $L$ are integers, we can solve for them using a grid search, and if this is too slow, a randomized search (\cite{bergstra2012random}). The only missing piece is $F$. 

To solve for $F$, we don't need to optimize over all the SWDFT coefficients. If a local periodic component is present, there would be large SWDFT coefficients in the corresponding frequency time-series (e.g. $a_{k, .}$). Therefore, we can restrict our optimization to only these coefficients:

\begin{eqnarray}
	\argmin_{S, L, F}( \argmin_{A, \phi} && \sum_{p = n - 1}^{N - 1} (b_{k, p} - \beta_1 C_{1, k, p} - \beta_2 C_{2, k, p})^2)
\end{eqnarray}

\noindent Which leads to the question, which $k$ should we consider?

\begin{enumerate}
	\item Select the $k$ with the largest squared modulus coefficient ($\max(|a_{k, p}|^2)$). 
	\item Select the $k$ that gives the largest {\bf reduction} in Mean Squared Error (MSE).
\end{enumerate}

Option 1 is fast and simple. If the signal has local periodic components, they would manifest as large squared modulus SWDFT coefficients, and the largest coefficient provides the strongest evidence of a local periodic signal. The potential for error in Option 1 occurs if we pick the wrong frequency, due to random noise at some other frequency. This error becomes more probable when the amplitude is small compared with the noise level. Of course, we could stabilize these estimates by selecting the $k$ with the largest, say 5, consecutive squared modulus coefficients (see Chapter 7 of \cite{okamura2011short}).
 
Option 2 searches {\bf all} combinations of Fourier frequencies ($k = 0, 1, \ldots n - 1$) and grid positions ($S$ and $L$), and selects the Fourier frequency $k$ with the best fit. By ``best fit'', we mean the Fourier frequency that leads to the largest reduction in MSE, compared with just fitting the mean. To clarify this, let $\text{MSE}_A$ be the mean-squared-error of simply fitting the mean, and let $\text{MSE}_B$ be the mean-squared-error of fitting our local signal model for parameters $S, L$, and $k$, and let $\text{MSE}_C$ be the change in MSE:

\begin{eqnarray}
	\text{MSE}_{A} &=& \sum_{p = n - 1}^{N - 1} (b_{k, p} - \bar{b}_{k, p})^2 \nonumber \\
	\text{MSE}_{B} &=& \sum_{p = n - 1}^{N - 1} (b_{k, p} - \hat{\beta}_1 C_{1, k, p} - \hat{\beta}_2 C_{2, k, p})^2\nonumber \\
	\text{MSE}_{C} &=& \text{MSE}_{B} - \text{MSE}_{A}
\end{eqnarray}

Option 2 minimizes $\text{MSE}_C$ opposed to just $\text{MSE}_B$ because $\text{MSE}_B$ may simply be small because the particular frequency time-series $a_{k, .}$ has small values. In this case, $\text{MSE}_B$ would be small even though these is no signal at this frequency. Selecting the frequency that minimizes $\text{MSE}_C$ helps alleviate this scaling issue. While option $2$ is more complete, it is also more computationally demanding, and doesn't necessarily optimize parameters other than $F$. However, the idea is useful, especially if we restricted possible frequencies to only the Fourier frequencies, which is common for DFT analysis. In this case, the ideas underlying Option $2$ would be useful. 

Either way, assuming we chose frequency $k^*$ correctly, the true frequency $f$ is {\bf closer} to frequency $k^*$ than any other $k \neq k^* = 0, 1, \ldots, n - 1$. This means we can restrict our search for $f$ to $[k^* - \frac{1}{2}, k^* + \frac{1}{2}]$, and solve numerically. The steps of our estimation procedure are summarized in Algorithm \ref{alg::estimation}.  

\begin{algorithm}
\caption{Estimate Parameters for One Local Periodic Signal from the SWDFT} 
\label{alg::estimation}
\SetAlgoLined
{\bf input: } {\bf b}, the SWDFT of a time-series \\
  	\vspace{.2em}
  - Select frequency $k$ time-series to search \\
  	\vspace{.2em}
  \For{Discrete Search over $S$ and $L$} {
  		\vspace{.5em}
  	\For{f in $[k - \frac{1}{2}, k + \frac{1}{2}]$} {
  		\vspace{.5em}
	  -	Create $C_{1, k, p}$ and $C_{2, k, p}$ \\
	  	\vspace{.2em}
	  - Estimate $\hat{\beta}_1$ $\hat{\beta}_2$ \\
	  	\vspace{.2em}
	  - Convert $\hat{\beta}_1, \hat{\beta}_2$ to $\hat{A}, \hat{\phi}$ \\
  	}	
  }
  {\bf output: } Parameters $\hat{S}, \hat{L}, \hat{A}, \hat{F}$, and $\hat{\phi}$ with lowest MSE \\
\end{algorithm}

\subsection{Estimation with Noise}
\label{sec:estimationplusnoise}
Until now, we only considered noiseless signals, since this helped us explain key features of the SWDFT. For practical time-series, we need to deal with noise. This section presents a very small and exploratory simulation study. Our goal is a better understanding of how noise effects the accuracy of our estimation procedure. We use the following model:

\begin{eqnarray}		
	y_t &=& A \cos(\frac{2 \pi F t}{N} + \phi) \mathbb{1}_{S, S + L - 1}  + \epsilon_t \nonumber \\
	&& t = 0, 1, \ldots, N - 1
\end{eqnarray}

\noindent For simplicity, assume the errors ${\bf \epsilon} = [\epsilon_1, \ldots, \epsilon_{N - 1}]$ are iid and $\epsilon_t \sim N(0, \sigma^2)$. 

Similar to the simulations in \cite{siegel1980testing}, we use the fact that signals with larger amplitudes are easier to detect. So, to determine how noise effects our estimation accuracy, the most important quantity is what we call the ``amplitude-to-noise'' ratio:

\begin{eqnarray}
	\text{Amplitude-to-Noise-Ratio} &=& \frac{A}{\sigma}
\end{eqnarray}

\noindent Our simulations hold amplitude ($A$) fixed and vary $\sigma$. In addition, our simulations include the effect of leakage and window size. We vary the following parameters in our simulations:

\begin{itemize}
	\item $n = 8, 16, 32$ 
	\item $\sigma = 0, .5, 1, 1.5, 2$ 
	\item $F = 8, 11$. 
\end{itemize}

\noindent And we keep the remaining parameters fixed:

\begin{itemize}
	\item $A = 1$
	\item $N = 64$
	\item $S = 17$
	\item $L = 31$
	\item $\phi = 1$
\end{itemize}

\noindent We fix $S$ and $L$ because, unless the local signal is near one of the ends of the time-series, the exact numerical positions of $S$ and $L$ are not critical. Finally, we fix $\phi$ because it wasn't the focus of this study (although we still want to estimate it well!). 

With this configuration, our estimation procedure takes approximately $30$ seconds to run for a single dataset. We ran each set of parameters $25$ times. The three tables in the appendix show the MSE of our parameter estimates. We conclude this section with a few key observations from the study. 

Figure \ref{fig:selectingk} shows the proportion of simulations where we estimated the frequency time-series to search correctly. As expected, when the signal gets noisier, we are less likely to search the correct SWDFT time-series. This is critical, because if we aren't searching the correct frequency, the parameter estimates are just fit to noise. 

\begin{figure}[ht]
	\centering
	\includegraphics[width = 14cm]{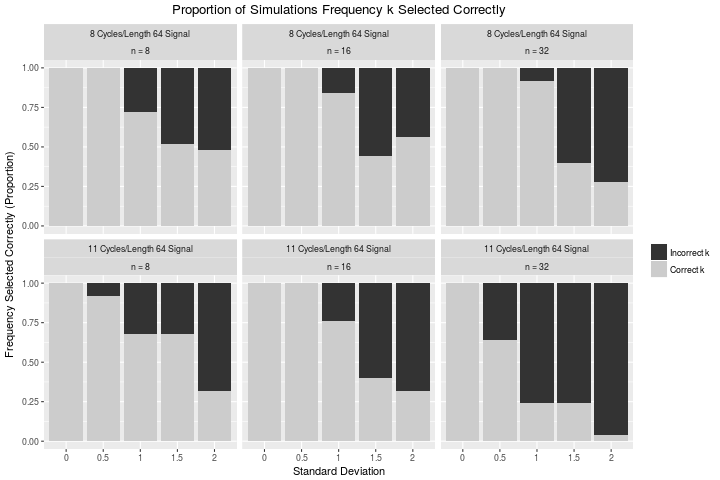} 
	\caption{Proportion of our cases where our estimation procedure selected the correct frequency. The top three panels have window sizes $n = 8, 16$, and $32$, and are based on signals with no leakage ($8$ complete cycles in a length $64$ time-series). The bottom three panels are the same window sizes as the top, but this time correspond to $11$ complete cycles in a length $64$ time-series. The x-axis is the five different values of $\sigma$ we used in our simulations.}
	\label{fig:selectingk}
\end{figure}

Figure \ref{fig:ahat} shows our estimates of the amplitude parameter ($\hat{A}$). Our estimation clearly gets worse as noise is added. However, an interesting point is that our amplitude estimates are biased upwards, since high noise leads to larger values in the original signal.

\begin{figure}[ht]
	\centering
	\includegraphics[width = 12cm]{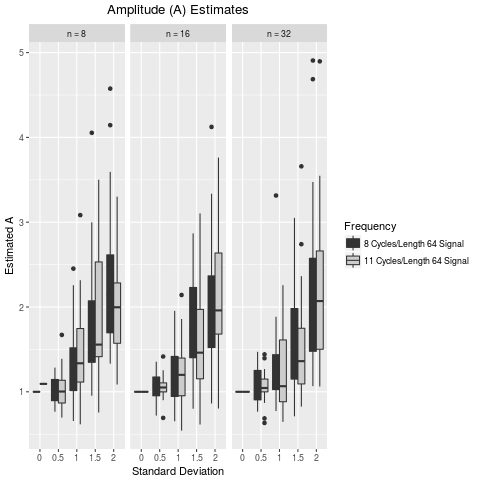} 
	\caption{Estimates of amplitude $\hat{A}$ parameter from our simulation study. The three panels correspond to window sizes $n = 8, 16$, and $32$, the black box-plots correspond to frequencies with no leakage ($8$ complete cycles in length $64$ time-series), and the gray box-plots correspond to frequencies with leakage ($11$ complete cycles).}
	\label{fig:ahat}
\end{figure}

\section{Discussion}
\label{sec:discussion}
This paper introduced the SWDFT as a tool for analyzing local periodic signals. We proposed a $5$-parameter model for a noiseless local periodic signal, and studied the SWDFT under different parameterizations of the model. Then we showed how the ideas underlying the study extend to $R > 1$ local periodic signals. Next, we proposed a simple estimation procedure for the least-squares parameters of the model using the SWDFT. We concluded with a very small Monte Carlo simulation study of the procedure's accuracy under different noise levels. 

The SWDFT is also useful for Exploratory Data Analysis (EDA). To see this, consider the two examples shown in Figure \ref{fig:sunspot_lynx}. The left panels show the time-series and corresponding SWDFT of Canadian Lynx trappings (\cite{campbell1977survey}), and the right panels show the same for annual sunspot numbers collected at the Swiss Federal Observatory from 1700-1988. Both data-sets are included by default in the {\tt R} language distribution, and can be accessed by {\tt data("lynx")} and {\tt data("sunspot.year")}. The SWDFT of the Canadian Lynx data shows consistently large energy at frequency time-series $3$ ($|a_{3, .}|^2$), which corresponds to a cycle every $\frac{32}{3} \approx 10.6$ years, since we use a window size of $n = 32$. The consistency of this cycle suggests this 10 year oscillation is stationary. The sunspot data is more elusive, but we see the largest consistent cycle at frequency time-series $6$, which also corresponds to $\frac{64}{6} \approx 10.6$ years, which suggests the well-known 11-year cycle. However, for the window positions indexed by years, 1800-1850, this frequency time-series gets smaller, and frequency time-series $1$, corresponding to $64$ years, gets larger. This suggests that longer term periodicity exists in the sunspot series, such as the cycles discovered by \cite{ohtomo1994new}. More complete analysis of these two time-series can be found at \cite{campbell1977survey,stenseth1997population,ohtomo1994new}. 

\begin{figure}[ht]
	\centering
	\includegraphics[width = 16cm]{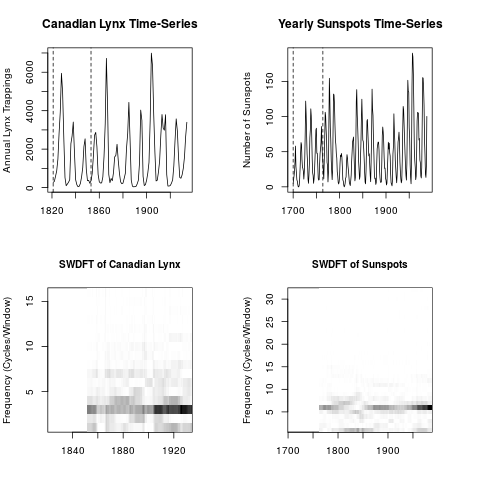} 
	\caption{Two examples of the SWDFT for Exploratory Data Analysis (EDA). Top Left: Time-series of Annual Canadian Lynx trappings. The two dashed lines show the input data for the first length $n = 32$ window position. Bottom Left: SWDFT of Canadian Lynx time-series for window size $n = 32$. The SWDFT shows the squared modulus SWDFT coefficients: $|a_{k, p}|^2$, $k=1,\ldots,16$.  Top right: Time-Series of yearly sunspot numbers from 1700-1988 at the Swiss Federal Observatory. The two dashed lines show the size of the first length $n = 64$ window position. Bottom Right: SWDFT of Yearly Sunspots time-series with window size $n = 64$.}
	\label{fig:sunspot_lynx}
\end{figure}

All computations in the paper used the \texttt{R} package {\tt swdft}, available online at \url{https://github.com/leerichardson/swdft}.

\clearpage 
\bibliographystyle{apa}
\bibliography{references}

\clearpage 	
\appendix
\section{Details on Estimation Procedure used in Simulation Study}
We include a few minor details important to our estimation and simulations here. 

\begin{itemize}
	\item The first issue comes from searching frequency time-series $a_{\frac{n}{2}, .} = [a_{\frac{n}{2}, n - 1}, a_{\frac{n}{2}, n}, \ldots, a_{\frac{n}{2}, N - 1}]$. This frequency time-series is problematic because it rapidly oscillates back and forth, since $\cos(\frac{2 \pi (n / 2)t}{n}) = \cos(\pi t)$. For estimation purposes, this is an issue because either the real or imaginary values can {\bf all} be extremely small, which means estimates for $\hat{\beta}_1$ and $\hat{\beta}_2$ will be extremely large to compensate. These large coefficients correspond to large amplitude estimates, since our conversion is $\hat{A} = \sqrt{\beta_1^2 + \beta_2^2)}$. To give an idea of how problematic this was, some of our simulations estimated $\hat{A} > 1000$, even though the true amplitude was $A = 1.$

	We addressed this by not allowing our estimation procedure to search frequency $\frac{n}{2}$. In practice, if this is the frequency of interest, the window size should be increased. 
	\item To numerically solve for $f$, we use the \textit{optimize} R function, which uses ``a combination of golden section search and successive parabolic interpolation'' (\cite{rcoreteam}).
	\item We used Option 1 described in Section \ref{sec:optfk} to select which frequency-time series to search. 
	\item In Section \ref{sec:linearizing}, we said that $\hat{\beta}_1$ and $\hat{\beta}_2$ could be estimated using the real parts of the equation. Of course, they can also be estimated using the imaginary parts. In our simulation, we arbitrarily chose to use the real-parts of the equation.
	\item We restricted our search for the parameter $L$ to $L \geq 8$. While the value of $8$ is arbitrary, we did this because the smaller the signal we are searching for, the more likely it would appear simply as noise, and we found $8$ a good cut-off for avoiding this problem. 
\end{itemize}

\begin{table}[ht]
\begin{tabular}{l*{5}{c}}
	Estimates of $A$, $n = 8$ & \multicolumn{5}{c}{Standard Deviation ($\sigma$)} \\ 
	\cmidrule(lr){2-6} 
		Frequency & 0 & .5 & 1 & 1.5 & 2 \\
	\midrule 
		8 Cycles/Length 64 Signal & 0.00 & 0.02 & 0.28 & 1.20 & 2.39 \\ 
		11 Cycles/Length 64 Signal & 0.01 & 0.05 & 0.52 & 1.24 & 1.14 \\ 
	\bottomrule
\end{tabular}
\vspace{2em}

\begin{tabular}{l*{5}{c}}
	Estimates of $S$, $n = 8$ & \multicolumn{5}{c}{Standard Deviation ($\sigma$)} \\ 
	\cmidrule(lr){2-6} 
		Frequency & 0 & 1 & 2 & 4 & 8\\
	\midrule 
		8 Cycles/Length 64 Signal & 0.00 & 4.80 & 90.00 & 182.72 & 278.20 \\ 
		11 Cycles/Length 64 Signal & 9.00 & 22.84 & 112.16 & 214.48 & 157.36 \\ 
	\bottomrule 
\end{tabular}
\vspace{2em}

\begin{tabular}{l*{5}{c}}
	Estimates of $L$, $n = 8$ & \multicolumn{5}{c}{Standard Deviation ($\sigma$)} \\ 
	\cmidrule(lr){2-6} 
	Frequency & 0 & 1 & 2 & 4 & 8\\
	\midrule 
		8 Cycles/Length 64 Signal & 0.00 & 14.08 & 164.64 & 210.04 & 228.56 \\ 
		11 Cycles/Length 64 Signal & 25.00 & 59.84 & 161.88 & 250.96 & 165.12 \\ 
  	\bottomrule 
\end{tabular}
\vspace{2em}

\begin{tabular}{l*{5}{c}}
	Estimates of $f$, $n = 8$ & \multicolumn{5}{c}{Standard Deviation ($\sigma$)} \\ 
	\cmidrule(lr){2-6} 
	Frequency & 0 & 1 & 2 & 4 & 8\\
	\midrule 
		8 Cycles/Length 64 Signal & 0.00 & 0.00 & 0.75 & 1.20 & 1.45 \\ 
		11 Cycles/Length 64 Signal & 0.00 & 0.00 & 0.37 & 0.51 & 1.37 \\ 		
  	\bottomrule 
\end{tabular}
\vspace{2em}

\begin{tabular}{l*{5}{c}}
	Estimates of $\phi$, $n = 8$ & \multicolumn{5}{c}{Standard Deviation ($\sigma$)} \\ 
	\cmidrule(lr){2-6} 
	Frequency & 0 & 1 & 2 & 4 & 8\\
	\midrule 
		8 Cycles/Length 64 Signal & 0.00 & 1.25 & 3.00 & 7.96 & 7.69 \\ 
		11 Cycles/Length 64 Signal & 0.12 & 3.12 & 3.41 & 7.09 & 8.03 \\ 
  	\bottomrule 
\end{tabular}

\begin{tabular}{l*{5}{c}}
	Estimates of $k$, $n = 8$ & \multicolumn{5}{c}{Standard Deviation ($\sigma$)} \\ 
	\cmidrule(lr){2-6} 
	Frequency & 0 & 1 & 2 & 4 & 8\\
	\midrule 
		8 Cycles/Length 64 Signal & 1.00 & 1.00 & 0.72 & 0.52 & 0.48 \\ 
		11 Cycles/Length 64 Signal & 1.00 & 0.92 & 0.68 & 0.68 & 0.32 \\ 
  	\bottomrule 
\end{tabular}

\label{tab:results_n8}
\caption{These six tables show the accuracy of our estimation procedure in our simulation study when $n = 8$. The first $5$ tables give the Mean Squared Error (MSE) for the parameters $A$, $S$, $L$, $f$, and $\phi$, and the $6^{th}$ table gives the fraction of simulations out of 25 where we searched the correct frequency time-series.}
\end{table}

\begin{table}[ht]
\label{tab:results_n16}
\begin{tabular}{l*{5}{c}}
	Estimates of $A$, $n = 16$ & \multicolumn{5}{c}{Standard Deviation ($\sigma$)} \\ 
	\cmidrule(lr){2-6} 
	Frequency & 0 & 1 & 2 & 4 & 8\\
	\midrule 
		8 Cycles/Length 64 Signal & 0.00 & 0.03 & 0.13 & 0.94 & 1.54 \\ 
		11 Cycles/Length 64 Signal & 0.00 & 0.02 & 0.20 & 0.85 & 1.68 \\ 
	\bottomrule
\end{tabular}
\vspace{2em}

\begin{tabular}{l*{5}{c}}
	Estimates of $S$, $n = 16$ & \multicolumn{5}{c}{Standard Deviation ($\sigma$)} \\ 
	\cmidrule(lr){2-6} 
	Frequency & 0 & 1 & 2 & 4 & 8\\
	\midrule
		8 Cycles/Length 64 Signal & 0.00 & 3.76 & 56.20 & 151.72 & 136.24 \\ 
		11 Cycles/Length 64 Signal & 0.00 & 5.40 & 51.04 & 127.64 & 185.28 \\ 
	 \bottomrule 
\end{tabular}
\vspace{2em}

\begin{tabular}{l*{5}{c}}
	Estimates of $L$, $n = 16$ & \multicolumn{5}{c}{Standard Deviation ($\sigma$)} \\ 
	\cmidrule(lr){2-6} 
	Frequency & 0 & 1 & 2 & 4 & 8\\
	\midrule 
		8 Cycles/Length 64 Signal & 0.00 & 13.40 & 88.96 & 227.48 & 287.96 \\ 
		11 Cycles/Length 64 Signal & 0.00 & 12.12 & 160.36 & 218.48 & 152.00 \\ 
  	\bottomrule 
\end{tabular}
\vspace{2em}

\begin{tabular}{l*{5}{c}}
	Estimates of $f$, $n = 16$ & \multicolumn{5}{c}{Standard Deviation ($\sigma$)} \\ 
	\cmidrule(lr){2-6} 
	Frequency & 0 & 1 & 2 & 4 & 8\\
	\midrule 
		8 Cycles/Length 64 Signal & 0.00 & 0.00 & 1.73 & 6.36 & 2.69 \\ 
		11 Cycles/Length 64 Signal & 0.00 & 0.00 & 1.56 & 2.81 & 3.67 \\ 
  	\bottomrule 
\end{tabular}
\vspace{2em}

\begin{tabular}{l*{5}{c}}
	\toprule 
	Estimates of $\phi$, $n = 16$ & \multicolumn{5}{c}{Standard Deviation ($\sigma$)} \\ 
	\cmidrule(lr){2-6} 
	Frequency & 0 & 1 & 2 & 4 & 8\\
	\midrule 
		8 Cycles/Length 64 Signal & 0.00 & 1.21 & 6.28 & 10.06 & 8.02 \\ 
		11 Cycles/Length 64 Signal & 0.00 & 1.62 & 5.87 & 7.81 & 6.98 \\ 
  	\bottomrule 
\end{tabular}

\begin{tabular}{l*{5}{c}}
	Estimates of $k$, $n = 16$ & \multicolumn{5}{c}{Standard Deviation ($\sigma$)} \\ 
	\cmidrule(lr){2-6} 
	Frequency & 0 & 1 & 2 & 4 & 8\\
	\midrule 
		8 Cycles/Length 64 Signal & 1.00 & 1.00 & 0.84 & 0.44 & 0.56 \\ 
		11 Cycles/Length 64 Signal & 1.00 & 1.00 & 0.76 & 0.40 & 0.32 \\ 
  	\bottomrule 
\end{tabular}

\caption{These six tables show the accuracy of our estimation procedure in our simulation study when $n = 16$. The first $5$ tables give the Mean Squared Error (MSE) for the parameters $A$, $S$, $L$, $f$, and $\phi$, and the $6^{th}$ table gives the fraction of simulations out of 25 where we searched the correct frequency time-series.}
\end{table}

\begin{table}[ht]
\label{tab:results_n32}
\begin{tabular}{l*{5}{c}}
	Estimates of $A$, $n = 32$ & \multicolumn{5}{c}{Standard Deviation ($\sigma$)} \\ 
	\cmidrule(lr){2-6} 
	Frequency & 0 & 1 & 2 & 4 & 8\\
	\midrule 
		8 Cycles/Length 64 Signal & 0.00 & 0.05 & 0.36 & 0.67 & 2.66 \\ 
		11 Cycles/Length 64 Signal & 0.00 & 0.04 & 0.29 & 0.66 & 2.24 \\ 
	\bottomrule
\end{tabular}
\vspace{2em}

\begin{tabular}{l*{5}{c}}
	Estimates of $S$, $n = 32$ & \multicolumn{5}{c}{Standard Deviation ($\sigma$)} \\ 
	\cmidrule(lr){2-6} 
	Frequency & 0 & 1 & 2 & 4 & 8\\
	\midrule
		8 Cycles/Length 64 Signal & 0.00 & 15.16 & 59.48 & 111.08 & 114.36 \\ 
		11 Cycles/Length 64 Signal & 0.00 & 35.64 & 70.28 & 139.64 & 112.32 \\ 
	 \bottomrule
\end{tabular}
\vspace{2em}

\begin{tabular}{l*{5}{c}}
	Estimates of $L$, $n = 32$ & \multicolumn{5}{c}{Standard Deviation ($\sigma$)} \\ 
	\cmidrule(lr){2-6} 
	Frequency & 0 & 1 & 2 & 4 & 8\\
	\midrule 
		8 Cycles/Length 64 Signal & 0.00 & 29.88 & 103.76 & 181.72 & 211.88 \\ 
		11 Cycles/Length 64 Signal & 0.00 & 35.28 & 118.52 & 201.76 & 111.04 \\ 
  	\bottomrule 
\end{tabular}
\vspace{2em}

\begin{tabular}{l*{5}{c}}
	Estimates of $f$, $n = 32$ & \multicolumn{5}{c}{Standard Deviation ($\sigma$)} \\ 
	\cmidrule(lr){2-6} 
	Frequency & 0 & 1 & 2 & 4 & 8\\
	\midrule 
		8 Cycles/Length 64 Signal & 0.00 & 0.02 & 3.62 & 22.15 & 25.80 \\ 
		11 Cycles/Length 64 Signal & 0.00 & 0.01 & 8.46 & 17.75 & 29.00 \\ 
  	\bottomrule 
\end{tabular}
\vspace{2em}

\begin{tabular}{l*{5}{c}}
	\toprule 
	Estimates of $\phi$, $n = 32$ & \multicolumn{5}{c}{Standard Deviation ($\sigma$)} \\ 
	\cmidrule(lr){2-6} 
	Frequency & 0 & 1 & 2 & 4 & 8\\
	\midrule 
		8 Cycles/Length 64 Signal & 0.00 & 2.00 & 5.82 & 7.49 & 7.49 \\ 
		11 Cycles/Length 64 Signal & 0.00 & 2.89 & 5.92 & 4.80 & 7.15 \\ 
  	\bottomrule 
\end{tabular}

\begin{tabular}{l*{5}{c}}
	Estimates of $k$, $n = 32$ & \multicolumn{5}{c}{Standard Deviation ($\sigma$)} \\ 
	\cmidrule(lr){2-6} 
	Frequency & 0 & 1 & 2 & 4 & 8\\
	\midrule 
		8 Cycles/Length 64 Signal & 1.00 & 1.00 & 0.92 & 0.40 & 0.28 \\ 
		11 Cycles/Length 64 Signal & 1.00 & 0.64 & 0.24 & 0.24 & 0.04 \\ 
  	\bottomrule 
\end{tabular}

\caption{These six tables show the accuracy of our estimation procedure in our simulation study when $n = 32$. The first $5$ tables give the Mean Squared Error (MSE) for the parameters $A$, $S$, $L$, $f$, and $\phi$, and the $6^{th}$ table gives the fraction of simulations out of 25 where we searched the correct frequency time-series.}
\end{table}

\end{document}